\documentclass[letterpaper, 10 pt, conference]{ieeeconf}  

\IEEEoverridecommandlockouts                              

\usepackage{amsmath}

\usepackage{amsthm}
\usepackage{amssymb}
\usepackage{graphicx}
\usepackage{bm}
\usepackage{algorithmic,algorithm}
\usepackage[noadjust]{cite}
\usepackage[font=small]{caption}
\usepackage{subcaption}
\usepackage{pgfplots}
\usetikzlibrary{decorations.markings}
\pgfplotsset{
    compat=newest,
}

\DeclareMathOperator*{\argmax}{arg\,max}

\DeclareMathOperator{\atantwo}{atan2}

\DeclareMathOperator{\dist}{\mathrm{dist}}

\newcommand{\x}{{\bm{x}}}

\newcommand{\p}{{\bm{p}}}

\newcommand{\mbx}{\mathbf{x}}

\newcommand{\mbf}{\mathbf{f}}

\newcommand{\cA}{\mathcal{A}}
\newcommand{\cB}{\mathcal{B}}

\newcommand{\cR}{\mathcal{R}}

\newcommand{\cT}{\mathcal{T}}

\newcommand{\cX}{\mathcal{X}}

\newcommand{\bR}{\mathbb{R}}

\newcommand{\mT}{\mathfrak{T}}

\theoremstyle{definition}
\newtheorem{remark}{Remark}

\newtheorem{theorem}{Theorem}

\usepackage{xcolor}

\renewcommand{\baselinestretch}{0.932}

\title{\LARGE \bf
Two-Player Reconnaissance Game with Half-Planar Target and Retreat Regions
}

\author{Yoonjae Lee \and Efstathios Bakolas \thanks{Y. Lee (PhD student) and E. Bakolas (Associate Professor) are with the Department of Aerospace Engineering
and Engineering Mechanics, The University of Texas at Austin,
Austin, Texas 78712-1221, USA, Emails: yol033@utexas.edu; bakolas@austin.utexas.edu}}

\begin{document}

\maketitle
\thispagestyle{empty}
\pagestyle{empty}

\begin{abstract}

This paper is concerned with the reconnaissance game that involves two mobile agents: the Intruder and the Defender. The Intruder is tasked to reconnoiter a territory of interest (target region) and then return to a safe zone (retreat region), where the two regions are disjoint half-planes, while being chased by the faster Defender. This paper focuses on the scenario where the Defender is not guaranteed to capture the Intruder before the latter agent reaches the retreat region. The goal of the Intruder is to minimize its distance to the target region, whereas the Defender's goal is to maximize the same distance. The game is decomposed into two phases based on the Intruder's myopic goal. The complete solution of the game corresponding to each phase, namely the Value function and state-feedback equilibrium strategies, is developed in closed-form using differential game methods. Numerical simulation results are presented to showcase the efficacy of our solutions.

\end{abstract}

\section{Introduction} \label{sec:intro}

Pursuit-evasion and target defense games have received a significant amount of attention due to their close connection with a wide range of applications in, for example, aerospace, military, and robotics. Not many attempts, however, have yet been made to address hybrid problems that possess both pursuit-evasion and target defense aspects, such as aerial reconnaissance and coast guarding problems. In these problems, the goal of an intruding agent is not necessarily to attack or reach a target, but to perform reconnaissance tasks in the vicinity of it and then escape to a safe zone before being neutralized by its opponent. This paper presents how to formulate this type of problems as a two-staged differential game and find equilibrium (i.e., worst-case) control policies for the agents involved.

\textit{Literature review}: The study of pursuit-evasion games was initiated by Isaacs in his seminal work \cite{isaacs1965differential}. Recent advances in the field of pursuit-evasion games are summarized in \cite{weintraub2020introduction}. One remarkable variant of pursuit-evasion games is the so-called target guarding game, the two-player version of which was first discussed in \cite{isaacs1965differential}. For the past few years, target guarding games with a stationary target or territory have been extensively studied. The related work in the literature can be categorized based on the shape of the target under consideration, which can be a point \cite{li2011defending,chipade2019multiplayer}, a line \cite{yan2018reach,von2020multiple,garcia2020multiple}, a circle \cite{bajaj2019dynamic}, or an arbitrary convex set \cite{fu2020guarding,lee2021optimal,lee2021guarding,fu2023justification}. One class of target guarding games that has recently attracted the attention of many researchers is the so-called perimeter-defense game \cite{shishika2018local,von2022circular}, which involves a defender(s) whose state is constrained along the boundary or perimeter of a target. The turret-defense game, which is a special case that has a circular target, has recently been studied in \cite{von2022turret}. Another remarkable class of target guarding games is the active target defense game, which involves a mobile target \cite{das2022guarding} or a maneuverable target (i.e., an evader)~\cite{liang2019differential,8340791,liang2022reconnaissance}.

In target guarding games, an intruder aims at reaching a target or, if capture is unavoidable, approach the target as close as possible until the time of capture. This problem formulation is relevant to the problem considered in this paper, namely the reconnaissance game. Similar to target guarding games, the reconnaissance game involves an intruder whose goal is to minimize its distance to the target. The difference, however, is that the intruder in the latter game has an additional goal which is to enter a safe zone before being captured by the defender. The fact that the intruder has dual goals makes the game difficult to analyze with the classical differential game approach. To our best knowledge, \cite{plante1972reconnaissance} was the first work that studied an open-loop solution of the reconnaissance game with two agents and a point target using the method of game decomposition. This method has later been revisited for many other problems that have more than one stages, such as the fixed-course target observation problem \cite{10018228}, the engage or retreat game \cite{fuchs2016generalized}, and the capture-the-flag game \cite{huang2014automation,garcia2018capture}. Note that the key difference between the capture-the-flag game and the reconnaissance game is that the latter game always terminates with the intruder entering the safe zone, whereas the former game has no such hard terminal constraint.


\textit{Statement of contributions}: The main contributions of this paper are as follows. Compared to \cite{plante1972reconnaissance} in which only the point target case was considered and an open-loop solution was developed in part based on iterative search methods, we develop the complete solution for the reconnaissance game involving two agents and half-planar target and retreat regions completely analytically. In particular, our method yields the solution to the game, namely the Value function and state-feedback equilibrium strategies of the game, in closed-form instead of relying on algorithmic or numerical methods as in \cite{plante1972reconnaissance}. Furthermore, we rigorously verify the validity of our solution using the Hamilton-Jacobi-Isaacs (HJI) equation, a step that has also been absent in \cite{plante1972reconnaissance}. Lastly, using the analytical expression of the Value function, we show how to divide the state space of the game into the winning sets of each agent and also how to characterize the barriers that demarcate these winning sets.

\textit{Outline}: The rest of the paper is structured as follows. In Section \ref{sec:probform}, the reconnaissance game is formulated and decomposed into two phases. In Sections \ref{sec:escapegame} and \ref{sec:targetgame}, the solution of the game corresponding to each phase is developed based on differential game methods. In Section \ref{sec:simresults}, simulation results are presented. Finally, in Section \ref{sec:conclusions}, concluding remarks are provided.

\section{Problem Formulation} \label{sec:probform}

In this section, the two-player reconnaissance game that takes place in $\bR^2$ is formulated. The game involves two mobile agents, the Intruder ($I$) and the Defender ($D$), whose equations of motion are given as
\begin{align}
    \dot x_I &= \cos \phi, & x_I(0) &= x_I^0, \label{eq:xI}
    \\
    \dot y_I &= \sin \phi, & y_I(0) &= y_I^0, \label{eq:yI}
    \\
    \dot x_D &= \alpha \cos \psi, & x_D(0) &= x_D^0, \label{eq:xD}
    \\
    \dot y_D &= \alpha \sin \psi, & y_D(0) &= y_D^0, \label{eq:yD}
\end{align}
 where $[x_I,y_I]^\top \in \bR^2$, $[x_I^0,y_I^0]^\top \in \bR^2 \backslash (\cT \cup \cR)$, and $\phi \in [-\pi,\pi]$ (resp., $[x_D,y_D]^\top \in \bR^2$, $[x_D^0,y_D^0]^\top \in \bR^2 \backslash \cR$, and $\psi \in [-\pi,\pi]$) denote the state/position, initial state/position, and control input/heading angle of the Intruder (resp., Defender), respectively; note that the sets $\cT$ and $\cR$ will be defined shortly. The scalar $\alpha$ denotes the speed of the Defender or the speed ratio, where $\alpha > 1$ (in other words, the Defender is faster than the Intruder). The dynamics of the game can be written as
\begin{align}
    \dot\mbx = \mbf(\mbx,\phi,\psi), \qquad \mbx(0) = \mbx^0,
\end{align}
where $\mbx = [x_I,y_I,x_D,y_D]^\top \in \bR^4$ is the game state, $\mbx^0 = [x_I^0,y_I^0,x_D^0,y_D^0]^\top \in \bR^4$ is the initial game state, and $\mbf : \bR^4 \times [-\pi,\pi] \times [-\pi,\pi] \rightarrow \bR^4$ is the (continuously differentiable) vector field of the game dynamics.

In this game, the Intruder is interested in reconnoitering (i.e., minimizing its distance to) the \textit{target region} $\cT$:
\begin{align}
    \cT = \left\{ [x,y]^\top \in \bR^2 : y \geq l \right\},
\end{align}
where $l > 0$. After performing the reconnaissance task, the same agent is required to return to the \textit{retreat region} $\cR$:
\begin{align} \label{eq:retreatregion}
    \cR = \left\{ [x,y]^\top \in \bR^2 : y \leq 0 \right\}.
\end{align}
By construction, $\cT$ and $\cR$ are closed half-spaces in $\bR^2$ (i.e., half-planes) that are disjoint. The Intruder is said to win the game if it enters $\cR$, which is equivalent to say, if the game state reaches the terminal manifold $\mT_r$:
\begin{align}
    \mT_r &= \left\{ \mbx \in \bR^4 : y_I \leq 0 \right\}. \label{eq:Tr}
\end{align}
The Defender, on the other hand, is considered to win if it captures (with zero capture radius) the Intruder outside $\cR$. The terminal manifold for capture, $\mT_c$, is defined as
\begin{align}
    \mT_c = \left\{ \mbx \in \bR^4 \backslash \mT_r : (x_I-x_D)^2+(y_I-y_D)^2 = 0 \right\}. \label{eq:Tc}
\end{align}
Note from \eqref{eq:Tr} and \eqref{eq:Tc} that the Intruder is considered to win if it is ``captured'' on the boundary of $\cR$.

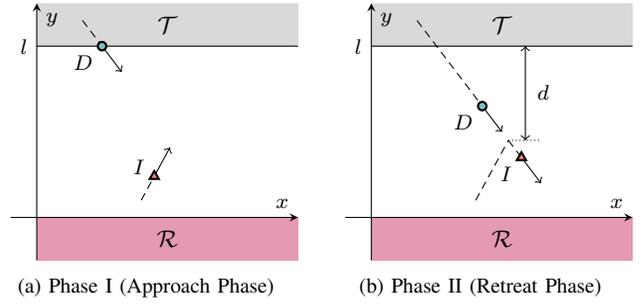
\begin{figure}
    \centering
    \subcaptionbox{Phase I (Approach Phase) \label{fig:1a}}{
    \begin{tikzpicture}
    \begin{axis}[
        font=\footnotesize,
        scatter/classes={
        a={mark=*,draw=black,thick,fill=teal!50!white,scale=0.8},
        b={mark=triangle*,draw=black,thick,fill=red!50!white,scale=1}
        },
        height=5cm, width=5.4cm,
        axis on top,
        xlabel=$x$, ylabel=$y$,
        xlabel style={anchor=south},
        ylabel style={anchor=south},
        axis x line=center,
        axis y line=center,
        xmin=-0.2, xmax=2,
        ymin=-0.25, ymax=1.25,
        xtick=\empty, ytick=\empty,
        ]
        \addplot [
        scatter,
        only marks,
        scatter src=explicit symbolic,
        ] coordinates {
        (0.5,1.0) [a]
        (0.9,0.24) [b]
        };
        \addplot [
        fill=purple!40!white,
        draw=none
        ] coordinates {
        (0,0) (2,0) (2,-0.25) (0,-0.25)
        }; 
        \addplot [
        fill=lightgray!60,
        draw=none
        ] coordinates {
        (0,1) (2,1) (2,1.25) (0,1.25)
        };  
        \addplot [
        black,
        domain=0:2,
        samples=10,
        ]
        {1};
        \draw[->](0.9,0.24)--(1.02,0.4080); 
        \draw[->](0.5,1.0)--(0.65,0.85); 
        \node [anchor=center,font=\small] at (1,1.125) {$\cT$}; 
        \node [anchor=center,font=\small] at (1,-0.125) {$\cR$}; 
        \node [anchor=north east,font=\footnotesize] at (0.5,1) {$D$}; 
        \node [anchor=south east,font=\footnotesize] at (0.9,0.2) {$I$}; 
        \draw[densely dashed,black](0.8,0.1)--(0.9,0.24); 
        \draw[densely dashed,black](0.35,1.15)--(0.5,1.0); 
        \node [anchor=east,font=\footnotesize] at (0,1) {$l$};
        \end{axis}
    \end{tikzpicture}
    }
    ~
    \subcaptionbox{Phase II (Retreat Phase) \label{fig:1b}}{
    \begin{tikzpicture}
    \begin{axis}[
        font=\footnotesize,
        scatter/classes={
        a={mark=*,draw=black,thick,fill=teal!50!white,scale=0.8},
        b={mark=triangle*,draw=black,thick,fill=red!50!white,scale=1,}
        },
        height=5cm, width=5.4cm,
        axis on top,
        xlabel=$x$, ylabel=$y$,
        xlabel style={anchor=south},
        ylabel style={anchor=south},
        axis x line=center,
        axis y line=center,
        xmin=-0.2, xmax=2,
        ymin=-0.25, ymax=1.25,
        xtick=\empty, ytick=\empty,
        ]
        \addplot [
        scatter,
        only marks,
        scatter src=explicit symbolic,
        ] coordinates {
        (0.85,0.65) [a]
        (1.15,0.35) [b]
        };
        \addplot [
        fill=purple!40!white,
        draw=none
        ] coordinates {
        (0,0) (2,0) (2,-0.25) (0,-0.25)
        };  
        \addplot [
        fill=lightgray!60,
        draw=none
        ] coordinates {
        (0,1) (2,1) (2,1.25) (0,1.25)
        };  
        \addplot [
        black,
        domain=0:2,
        samples=10,
        ]
        {1};
        \node [anchor=north east,font=\footnotesize] at (0.85,0.65) {$D$}; 
        \node [anchor=north east,font=\footnotesize] at (1.15,0.35) {$I$}; 
        \draw[->](1.15,0.35)--(1.3,0.2); 
        \draw[->](0.85,0.65)--(1.0,0.5); 
        \node [anchor=center,font=\small] at (1,1.125) {$\cT$}; 
        \node [anchor=center,font=\small] at (1,-0.125) {$\cR$}; 
        \draw[densely dashed,black](0.8,0.1)--(1.05,0.45); 
        \draw[densely dashed,black](1.05,0.45)--(1.15,0.35); 
        \draw[densely dashed,black](0.35,1.15)--(0.85,0.65); 
        \draw[<->](1.18,0.45)--(1.18,1);
        \draw[densely dotted](1.05,0.45)--(1.3,0.45);
        \node [anchor=west,font=\footnotesize] at (1.2,0.73) {$d$}; 
        \node [anchor=east,font=\footnotesize] at (0,1) {$l$};
        \end{axis}
    \end{tikzpicture}
    }
    \caption{Illustration of the two phases of the reconnaissance game. During Phase I, the Intruder (red triangle) approaches the target region $\cT$ (gray half-plane) to minimize its distance from the latter region as much as possible in the presence of the Defender (teal circle). When the Intruder realizes it should no longer proceed or otherwise capture may occur before entering the retreat region $\cR$ (purple half-plane), the game transitions to Phase II, in which the Intruder's goal is to arrive at the point in $\cR$ where the distance between the two agents is at its maximum. Note that the resulting payoff of the illustrated case is $d$.}
    \label{fig:recongame}
\end{figure}

Since the game of our interest requires that the Intruder return to $\cR$ before capture occurs, we will focus on the scenario where the game state ends up in $\mT_r$. In this case, the objective of the Intruder is to minimize its distance to $\cT$ (at some time instant during the game) as much as possible in the presence of the Defender that strives to maximize the same distance. The payoff functional of the reconnaissance game can thus be defined by
\begin{align} \label{eq:payoff}
    J &= \min_{t \in [0,t_f]} \dist([x_I(t),y_I(t)]^\top,\cT),
\end{align}
where $t_f = \inf \left\{ t \geq 0 : \mbx \in \mT_r \right\}$ is the final time and the function $\dist : \bR^2 \times 2^{\bR^2} \rightarrow [0,\infty)$ measures the (minimum) distance between a point and a subset of $\bR^2$. The Value of the game is
\begin{align} \label{eq:value}
    V(\mbx^0) = \min_{\phi(\cdot)}\max_{\psi(\cdot)} J,
\end{align}
where $\phi(\cdot)$ and $\psi(\cdot)$ denote the state-feedback strategies of each agent.

The goal of this paper is to find 1) the winning regions of each agent and 2) the Value function of the game, $V$, and the corresponding state-feedback equilibrium strategies, $\phi^\star(\cdot)$ and $\psi^\star(\cdot)$. By definition, the pair of equilibrium strategies must satisfy the saddle-point condition
\begin{align} \label{eq:sadpointcond}
    &J(\phi^\star(\cdot),\psi(\cdot);\mbx^0) \leq J(\phi^\star(\cdot),\psi^\star(\cdot);\mbx^0) \nonumber
    \\
    &\qquad\qquad \leq J(\phi(\cdot),\psi^\star(\cdot);\mbx^0),
\end{align}
for all possible $\phi(\cdot)$ and $\psi(\cdot)$.

Due to the non-integral payoff functional in \eqref{eq:payoff}, it is difficult to analyze the game in its current form. For this reason, as suggested in \cite{plante1972reconnaissance}, we will decompose the game into two phases based on the Intruder's myopic goal: approach or retreat. See Figure \ref{fig:recongame} for the illustration of Phase I (approach phase) and Phase II (retreat phase). In the following sections, we will solve Phase II first and then Phase I, as the solution of Phase I depends partially on that of Phase II.

\section{Phase II: Retreat Phase} \label{sec:escapegame}

Suppose at some time $t_s \in [0,t_f]$, the Intruder has arrived at the point closest to $\cT$. Simultaneously, the game transitions from Phase I to Phase II. In the latter phase, in order to reduce the chance of capture, the Intruder desires to maximize its distance from the Defender at the moment it reaches a point in $\cR$, which we refer to as a \textit{retreat point}. Conversely, the Defender attempts to increase the chance of capture by minimizing the same distance. The payoff functional of this phase is therefore defined as
\begin{align} \label{eq:J_II}
    J_\mathrm{II} = \sqrt{(x_I^f-x_D^f)^2+(y_I^f-y_D^f)^2},
\end{align}
where $\mbx^f = [x_I^f,y_I^f,x_D^f,y_D^f]^\top$ is the game state at $t=t_f$. Note that the game defined by the payoff \eqref{eq:J_II} and terminal manifold \eqref{eq:Tr} has in fact been studied thoroughly \cite{yan2018reach}, \cite{lee2021optimal}. For this reason, we will provide below the existing solution of this game, which is a slight modification of the results from the latter references.

\begin{theorem} \label{theo:1}
    Consider the game defined by the kinematics \eqref{eq:xI}-\eqref{eq:yD}, payoff \eqref{eq:J_II}, and terminal manifold \eqref{eq:Tr}. Define the function $\Delta : \bR^4 \times \bR^2 \rightarrow \bR$ as
    \begin{align} \label{eq:dist}
        \Delta(\mbx,\x) &= \sqrt{(x-x_D)^2+(y-y_D)^2} \nonumber
        \\
        &\quad - \alpha \sqrt{(x-x_I)^2+(y-y_I)^2},
    \end{align}
    where $\x = [x,y]^\top \in \bR^2$. Then, the function $V_\mathrm{II}(\mbx) : \bR^4 \rightarrow \bR$ where
    \begin{align} \label{eq:V_II}
        V_\mathrm{II}(\mbx) = \max_{\p \in \cR} \Delta(\mbx,\p)
    \end{align}
    is the Value function of the game over the domain
    \begin{align}
        \cX_\mathrm{II} = \left\{ \mbx \in \bR^4 : V_\mathrm{II}(\mbx) \geq 0 \right\},
    \end{align}
    and the corresponding state-feedback equilibrium strategies are given as
    \begin{align}
        \phi_\mathrm{II}^\star(\mbx) &= \atantwo(p_y^\star - y_I,p_x^\star-x_I), \label{eq:phi_II}
        \\
        \psi_\mathrm{II}^\star(\mbx) &= \atantwo(p_y^\star-y_D,p_x^\star-x_D), \label{eq:psi_II}
    \end{align}
    where $\p^\star = [p_x^\star,p_y^\star]^\top$ is the optimal retreat point given by
    \begin{align} \label{eq:pstar}
        \p^\star = \argmax_{\p \in \cR} \Delta(\mbx,\p).
    \end{align}
\end{theorem}
\begin{proof}
    See \cite{lee2021guarding}.
\end{proof}

\begin{remark}
    The function $\Delta$ computes the distance between the two agents when the Intruder reaches a retreat point, provided that both agents move along straight-line trajectories. $\cX_\mathrm{II}$ is the subset of the game state space in which the game corresponding to Phase II can take place (i.e., if $\mbx^0 \in \bR^4 \backslash \cX_\mathrm{II}$, the Defender can capture the Intruder before it reaches $\cR$). Lastly, although $V_\mathrm{II}$ in \eqref{eq:V_II} is written in general form, the boundary of $\cX_\mathrm{II}$ has an analytic expression as $\cR$ is a half-plane, for which the reader can refer to \cite{yan2018reach}.
\end{remark}

Before closing this section, we introduce a geometric definition that is relevant our problem. The set comprised of points that the Intruder can reach faster than or as fast as the Defender (i.e., the safe-reachable set of the Intruder) is a closed disk
\begin{align} \label{eq:apolcirc}
    \cA = \left\{ \x \in \bR^2 : \Delta(\mbx,\x) \geq 0 \right\},
\end{align}
whose center and radius are given by
\begin{align}
    \bm{c} &= \left[ \frac{\alpha^2 x_I - x_D}{\alpha^2-1},\frac{\alpha^2 y_I - y_D}{\alpha^2-1} \right]^\top, \label{eq:center}
    \\
    r &= \frac{\alpha}{\alpha^2-1} \sqrt{ (x_I-x_D)^2+(y_I-y_D)^2 }, \label{eq:radius}
\end{align}
respectively.

\section{Phase I: Approach Phase} \label{sec:targetgame}

In Phase I, the Intruder and the Defender compete to minimize and maximize the distance between $\cT$ and the Intruder's position at $t_s$, respectively. Since $\cT$ is a half-plane, the payoff functional of this phase can be defined by rewriting \eqref{eq:payoff} as
\begin{align} \label{eq:J_I}
    J_\mathrm{I} &= \Phi_\mathrm{I}(\mbx^s) = l - y_I^s,
\end{align}
where $\mbx^s = [x_I^s,y_I^s,x_D^s,y_D^s]^\top \in \cX_\mathrm{II}$ is the game state at $t=t_s$. Note that, without loss of generality, we allow \eqref{eq:J_I} to attain negative values, in which case the Intruder will enter $\cT$. The terminal constraint of this phase, which ensures safe retreat of the Intruder, is designed as follows:
\begin{align} \label{eq:psi_T}
    \Psi_\mathrm{I}(\mbx^s) &= \Delta(\mbx^s,\p) - \delta = 0,
\end{align}
where $\p = [p_x,p_y]^\top \in \cR$ is a retreat point and $\delta \geq 0$ is the distance between the two agents at the end of Phase II.

\begin{remark}
    To incorporate the solution of Phase II herein, $\p$ and $\delta$ in \eqref{eq:psi_T} should be defined as functions of $\mbx$. This approach, however, is intractable as these two variables admit no closed-form expressions. To overcome this difficulty, we let $\p$ and $\delta$ be parameters independent of $\mbx$. With this relaxation, one can develop an open-loop solution of the game for an arbitrary selection of these two parameters. Thereafter, the optimal pair that results in a solution in accordance with the solution of Phase II can be found.
\end{remark}

\subsection{Necessary Conditions for Optimality}

Let us first analyze the game defined by the payoff \eqref{eq:J_I} and terminal constraint \eqref{eq:psi_T} using the first-order necessary conditions for optimality. The Hamiltonian of this game is
\begin{align} \label{eq:hamiltonian}
    H &= \lambda_{x_I} \cos \phi + \lambda_{y_I} \sin \phi \nonumber
    \\
    &\quad + \alpha \left( \lambda_{x_D} \cos \psi + \lambda_{y_D} \sin \psi \right),
\end{align}
where $\lambda_{x_I},\lambda_{y_I},\lambda_{x_D}$, and $\lambda_{y_D}$ are the co-state variables. From \eqref{eq:hamiltonian}, one can infer that the open-loop equilibrium heading angles of the agents for Phase I, denoted by $\phi_\mathrm{I}^*$ and $\psi_\mathrm{I}^*$, satisfy
\begin{align}
    \cos \phi_\mathrm{I}^* &= -\frac{\lambda_{x_I}}{\sqrt{\lambda_{x_I}^2+\lambda_{y_I}^2}}, & \sin \phi_\mathrm{I}^* &= -\frac{\lambda_{y_I}}{\sqrt{\lambda_{x_I}^2+\lambda_{y_I}^2}}, \label{eq:optinput_I}
    \\
    \cos \psi_\mathrm{I}^* &= \frac{\lambda_{x_D}}{\sqrt{\lambda_{x_D}^2+\lambda_{y_D}^2}}, & \sin \psi_\mathrm{I}^* &= \frac{\lambda_{y_D}}{\sqrt{\lambda_{x_D}^2+\lambda_{y_D}^2}}, \label{eq:optinput_D}
\end{align}
for all $t \in [0,t_s]$. The equilibrium dynamics of the co-state variables are
\begin{align}
    \dot\lambda_{x_I} &= - \frac{\partial H}{\partial x_I} = 0, & \dot\lambda_{y_I} &= - \frac{\partial H}{\partial y_I} = 0,
    \\
    \dot\lambda_{x_D} &= - \frac{\partial H}{\partial x_D} = 0, & \dot\lambda_{y_D} &= - \frac{\partial H}{\partial y_D} = 0,
\end{align}
subject to the terminal boundary conditions
\begin{align}
    \lambda_{x_I}(t_s) &= \frac{\partial \Phi_\mathrm{I}(\mbx^s)}{\partial x_I} + \nu \frac{\partial \Psi_\mathrm{I}(\mbx^s)}{\partial x_I} = \alpha \nu \cos\theta_I^s, \label{eq:lamb_xI}
    \\
    \lambda_{y_I}(t_s) &= \frac{\partial \Phi_\mathrm{I}(\mbx^s)}{\partial y_I} + \nu \frac{\partial \Psi_\mathrm{I}(\mbx^s)}{\partial y_I} = -1 + \alpha \nu \sin\theta_I^s, \label{eq:lamb_yI}
    \\
    \lambda_{x_D}(t_s) &= \frac{\partial \Phi_\mathrm{I}(\mbx^s)}{\partial x_D} + \nu \frac{\partial \Psi_\mathrm{I}(\mbx^s)}{\partial x_D} = -\nu \cos\theta_D^s,
    \\
    \lambda_{y_D}(t_s) &= \frac{\partial \Phi_\mathrm{I}(\mbx^s)}{\partial y_D} + \nu \frac{\partial \Psi_\mathrm{I}(\mbx^s)}{\partial y_D} = -\nu \sin\theta_D^s,
\end{align}
where $\nu \in \bR$ is the Lagrange multiplier and
\begin{align}
    \theta_I^s &= \atantwo \left( p_y - y_I^s,p_x - x_I^s \right), \label{eq:theta_I}
    \\
    \theta_D^s &= \atantwo \left( p_y-y_D^s,p_x-x_D^s \right). \label{eq:theta_D}
\end{align}
Clearly, for all $t \in [0,t_s]$, we have $\lambda_{x_I}(t) = \lambda_{x_I}(t_s)$, $\lambda_{y_I}(t) = \lambda_{y_I}(t_s)$, $\lambda_{x_D}(t) = \lambda_{x_D}(t_s)$, and $\lambda_{y_D}(t) = \lambda_{y_D}(t_s)$. Using the fact that $H^*(t_s)=0$ yields
\begin{align}
    &H^*(t_s) = -\sqrt{ \alpha^2 \nu^2 - 2\alpha\nu\sin\theta_I^s + 1 } + \alpha |\nu| = 0 \nonumber
    \\
    &\Rightarrow \nu = \frac{\csc\theta_I^s}{2\alpha}. \label{eq:lagrange}
\end{align}
Substituting \eqref{eq:lagrange} into \eqref{eq:lamb_xI} and \eqref{eq:lamb_yI} gives
\begin{align}
    \lambda_{x_I} = \frac{1}{2} \cot\theta_I^s, \qquad \lambda_{y_I} = -\frac{1}{2}.
\end{align}
Consequently, \eqref{eq:optinput_I} and \eqref{eq:optinput_D} can be rewritten as
\begin{align}
    \cos \phi_\mathrm{I}^* &= \cos\theta_I^s, &\sin \phi_\mathrm{I}^* &= -\sin\theta_I^s,   \label{eq:phi_I_ol}
    \\
    \cos \psi_\mathrm{I}^* &= \cos\theta_D^s, & \sin \psi_\mathrm{I}^* &= \sin\theta_D^s. \label{eq:psi_I_ol}
\end{align}

\begin{remark} \label{rem:defopthead}
    From \eqref{eq:psi_II} and \eqref{eq:psi_I_ol}, one can deduce that, if $\p$ is the optimal retreat point yielded by \eqref{eq:pstar}, then the open-loop equilibrium heading angles of the Defender in Phases I and II are the same, meaning that the trajectory of the Defender is a straight line that connects its initial and terminal positions.
\end{remark}

\subsection{Optimal Selection of the Parameters} \label{sec:optselection}

The open-loop solution derived above is optimal only for the certain pair of $\p$ and $\delta$ selected. That is, there may exist a different pair that yields the smaller Value. It is therefore remained to show how $\p$ and $\delta$ should be selected in such a way that the ``optimal'' solution, which not only minimizes \eqref{eq:J_I} (from the perspective of the Intruder) but also aligns with the equilibrium solution for Phase II, can be induced. To answer this, let us define
\begin{align}
    X_I^0 &= p_x - x_I^0, \qquad Y_I^0 = p_y - y_I^0,
    \\
    \eta_D^0 &= \sqrt{(p_x-x_D^0)^2+(p_y-y_D^0)^2}.
\end{align}
Suppose $\p$ satisfies \eqref{eq:pstar} at $t_s$. Then, in light of Remark \ref{rem:defopthead}, we may write
\begin{align} \label{eq:tf}
    t_f = \frac{\eta_D^0 - \delta}{\alpha}.
\end{align}
To ensure $t_f \geq 0$, we need $\delta \in [0,\eta_D^0]$. From \eqref{eq:phi_II} and \eqref{eq:phi_I_ol}, one can easily deduce that the Intruder will move with constant speed along the $x$ direction for all $t \in [0,t_f]$, while the sign of its $y$-directional motion will switch at $t_s$. That is,
\begin{align}
    X_I^0 &= t_f \cos\theta_I^s \Rightarrow \cos\theta_I^s = \frac{X_I^0}{t_f}, \label{eq:costheta}
    \\
    Y_I^0 &= -t_s \sin\theta_I^s + (t_f - t_s) \sin\theta_I^s \nonumber
    \\
    &\quad \Rightarrow \sin\theta_I^s = \frac{Y_I^0}{t_f - 2 t_s}. \label{eq:sintheta}
\end{align}
By using the trigonometric identity  $\sin^2\theta_I^s+\cos^2\theta_I^s = 1$, we obtain
\begin{align}
      &t_s = \frac{t_f}{2} \left( 1 \pm \frac{Y_I^0}{\sqrt{t_f^2 - (X_I^0)^2}} \right).
\end{align}
Since the $y$-coordinate of the Intruder's position is always greater than zero while the game is ongoing (otherwise, $\mbx \in \mT_r$), it must always hold that $\sin\theta_I^s < 0$. Thus,
\begin{align} \label{eq:ts}
    t_s = \frac{t_f}{2} \left( 1 + \frac{Y_I^0}{\sqrt{t_f^2 - (X_I^0)^2}} \right).
\end{align}
Substituting \eqref{eq:ts} into \eqref{eq:sintheta} gives
\begin{align} \label{eq:sintheta_2}
    \sin\theta_I^s &= -\sqrt{1 - \left( \frac{X_I^0}{t_f} \right)^2}.
\end{align}
One can subsequently find the $y$-coordinate of the Intruder's position at $t_s$ using \eqref{eq:ts} and \eqref{eq:sintheta_2}. The Value of the game, which is the distance between $\cT$ and the Intruder's position at $t_s$, is therefore given by
\begin{align}
    V_\mathrm{I}(\mbx^0) &= l - y_I^0 - \frac{1}{2} \left( \sqrt{t_f^2 - (X_I^0)^2} + Y_I^0 \right). \label{eq:valfun_open}
\end{align}
As mentioned, our goal is to find the pair $\p^* = [p_x^*,p_y^*]^\top \in \cR$ and $\delta^* \in [0,\eta_D^0]$ that minimize \eqref{eq:valfun_open}. Observe first from \eqref{eq:retreatregion} that $\p^*$ will always lie on the boundary of $\cR$, that is, $p_y^* = 0$. In addition, to minimize \eqref{eq:valfun_open}, $t_f$ should be maximized, so $\delta^* = 0$; note that this value implies that both agents will simultaneously arrive at $\p^*$, which however still leads to the Intruder's win, as mentioned in Section \ref{sec:probform}. Now, let us rewrite \eqref{eq:valfun_open} evaluated at $p_y^*$ and $\delta^*$ as a function of $p_x$:
\begin{align} \label{eq:F}
    F(p_x) = l - \frac{1}{2} \left( y_I^0 + \Gamma^{\frac{1}{2}} \right),
\end{align}
where $\Gamma = \frac{1}{\alpha^2}(p_x-x_D^0)^2 - (p_x-x_I^0)^2 + \frac{1}{\alpha^2}(y_D^0)^2$. Here, we implicitly assume that the domain of $F$ is restricted to the interval of $p_x$ where $\Gamma \geq 0$, which turns out to be compact. The first and second derivatives of $F$ are derived as
\begin{align}
    \frac{\partial F(p_x)}{\partial p_x} &= -\frac{1}{2\Gamma^{\frac{1}{2}}} \left( \frac{1-\alpha^2}{\alpha^2} p_x + x_I^0 - \frac{1}{\alpha^2} x_D^0 \right), \label{eq:firstder}
    \\
    \frac{\partial^2 F(p_x)}{\partial p_x^2} &= \frac{1}{2\Gamma^{\frac{3}{2}}} \left( \frac{1-\alpha^2}{\alpha^2} p_x + x_I^0 - \frac{1}{\alpha^2} x_D^0 \right)^2 + \frac{\alpha^2-1}{2 \alpha^2 \Gamma^{\frac{1}{2}}} . \label{eq:secondder}
\end{align}
Since the right-hand side of \eqref{eq:secondder} is always nonnegative, one can immediately conclude that $F$ is convex. The $p_x^*$ that makes the first derivative of $F$ zero can easily be found from \eqref{eq:firstder} as $p_x^* = \frac{\alpha^2 x_I^0 - x_D^0}{\alpha^2-1}$. In summary, we have obtained
\begin{align} \label{eq:optpt}
    \p^* = \left[ \frac{\alpha^2 x_I^0 - x_D^0}{\alpha^2-1},0 \right]^\top, \quad \delta^* = 0.
\end{align}

\begin{remark} \label{rem:proj}
    The point $\p^*$ given in \eqref{eq:optpt} is the orthogonal projection of the center of the disk $\cA^0 = \{\x \in \bR^2 : \Delta(\mbx^0,\x) \geq 0\}$ onto $\cR$ (see \eqref{eq:center}).
\end{remark}

Due to space limitations, the rigorous proof for showing that $\p^*$ satisfies \eqref{eq:pstar} is left for future work. The sketch of this proof is that the boundary of $\cR$ is tangent to the disk $\cA^s = \{\x \in \bR^2 : \Delta(\mbx^s,\x) \geq 0\}$ at $\p^*$, that is, $\Delta(\mbx^s,\p^*) = 0$ and $\Delta(\mbx^s,\p) < 0$ for all $\p \in \cR \backslash \{\p^*\}$.

\subsection{State-Feedback Equilibrium Solution}

The next result to be presented below is the state-feedback equilibrium solution of the game corresponding to Phase I, which we develop based on the open-loop solution derived in the previous subsections.

\begin{theorem} \label{theo:2}
    Consider the function $V_\mathrm{I} : \cX_\mathrm{II} \rightarrow \bR$ where
    \begin{align} \label{eq:valfun}
        V_\mathrm{I}(\mbx) &= l - \frac{1}{2} \left( y_I + \sqrt{ \frac{(x_I-x_D)^2}{\alpha^2-1} + \frac{y_D^2}{\alpha^2} } \right).
    \end{align}
    Then, $V_\mathrm{I}$ is the Value function of the game defined by the kinematics \eqref{eq:xI}-\eqref{eq:yD}, payoff \eqref{eq:J_I}, and terminal manifold \eqref{eq:Tr}, over the domain
    \begin{align}
        \cX_\mathrm{I} = \left\{ \mbx \in \cX_\mathrm{II} : V_\mathrm{I}(\mbx) \geq 0 \right\},
    \end{align}
    and the corresponding state-feedback equilibrium strategies are given as
    \begin{align}
        \phi_\mathrm{I}^\star(\mbx) &= \atantwo\left( \sqrt{ 1-\Lambda^2 },\Lambda \right), \label{eq:phi_I}
        \\
        \psi_\mathrm{I}^\star(\mbx) &= \atantwo \left( -\sqrt{ 1-\alpha^2 \Lambda^2 },\alpha \Lambda \right), \label{eq:psi_I}
    \end{align}
    where $\Lambda = 1/\sqrt{\alpha^2 + \frac{(\alpha^2-1)^2}{\alpha^2} \frac{y_D^2}{(x_I-x_D)^2}}$.
\end{theorem}
\begin{proof}
    The expression of $V_\mathrm{I}$ in \eqref{eq:valfun} is obtained by substituting \eqref{eq:optpt} into \eqref{eq:F}, and the strategies in \eqref{eq:phi_I} and \eqref{eq:psi_I} by substituting \eqref{eq:optpt} into \eqref{eq:phi_I_ol}, \eqref{eq:psi_I_ol}, \eqref{eq:costheta}, and \eqref{eq:sintheta_2}. The partial derivatives of $V_\mathrm{I}$ with respect to $\mbx$ over $\cX_\mathrm{I}$ are found as
    \begin{align} \label{eq:dVdx}
        \frac{\partial V_\mathrm{I}}{\partial \mbx} = -\frac{1}{2} \left[
        \frac{\Lambda}{\sqrt{1-\Lambda^2}} ~~ 1 ~~ -\frac{\Lambda}{\sqrt{1-\Lambda^2}} ~~ \frac{1}{\alpha} \sqrt{ \frac{1-\alpha^2 \Lambda^2}{1-\Lambda^2} }
        \right]
        .
    \end{align}
    The HJI equation is given as \cite{isaacs1965differential}
    \begin{align} \label{eq:hji}
        0 = \frac{\partial V_\mathrm{I}}{\partial t} + \min_{\phi}\max_{\psi} \left\{ \frac{\partial V_\mathrm{I}}{\partial \mbx} \cdot \mbf(\mbx,\phi,\psi) \right\}.
    \end{align}
    Since $V_\mathrm{I}$ is not an explicit function of time, we have $\frac{\partial V_\mathrm{I}}{\partial t} = 0$. Substituting \eqref{eq:phi_I}, \eqref{eq:psi_I}, and \eqref{eq:dVdx} into \eqref{eq:hji} yields
    \begin{align}
        &\min_\phi \max_\psi \left\{ \frac{\partial V_\mathrm{I}}{\partial \mbx} \cdot \mbf(\mbx,\phi,\psi) \right\} \nonumber
        \\
        &= -\frac{1}{2}
        \begin{bmatrix}
            \frac{\Lambda}{\sqrt{1-\Lambda^2}}
            \\
            1
            \\
            -\frac{\Lambda}{\sqrt{1-\Lambda^2}}
            \\
            \frac{1}{\alpha} \sqrt{ \frac{1-\alpha^2 \Lambda^2}{1-\Lambda^2} }
        \end{bmatrix}^\top
        \begin{bmatrix}
            \Lambda
            \\
            \sqrt{1-\Lambda^2}
            \\
            \alpha^2 \Lambda
            \\
            -\alpha \sqrt{1-\alpha^2 \Lambda^2}
        \end{bmatrix} = 0.
    \end{align}
    Since $V_\mathrm{I}$ is continuously differentiable and satisfies the HJI equation over $\cX_\mathrm{I}$, it is the Value function of the game.
\end{proof}

\subsection{Level Sets of Value Function and State Space Partitioning} \label{sec:level}

Figure \ref{fig:value} illustrates the level sets of the Value function described in \eqref{eq:valfun} and the winning regions of each agent, for which we select $[x_D^0,y_D^0]^\top = [0,0.6]^\top$, $\alpha = 1.1$, and $\l = 1$). Therein, the position of the Defender is marked by a teal circle, and the sets partitioning the state space are given as
\begin{align}
    \bar\cX_\mathrm{I} &= \left\{ [x,y]^\top \in \bar\cX_\mathrm{II} : V_\mathrm{I}([x,y,x_D^0,y_D^0]^\top) \geq 0 \right\},
    \\
    \bar\cX_\mathrm{II} &= \left\{ [x,y]^\top \in \bR^2 : V_\mathrm{II}([x,y,x_D^0,y_D^0]^\top) \geq 0 \right\},
    \\
    \bar\cB_\mathrm{I} &= \left\{ [x,y]^\top \in \bar\cX_\mathrm{II} : V_\mathrm{I}([x,y,x_D^0,y_D^0]^\top) = 0 \right\},
    \\
    \bar\cB_\mathrm{II} &= \left\{ [x,y]^\top \in \bR^2 : V_\mathrm{II}([x,y,x_D^0,y_D^0]^\top) = 0 \right\}.
\end{align}
The black curves $\bar\cB_\mathrm{I}$ and $\bar\cB_\mathrm{II}$ are the so-called \textit{barriers}, which demarcate the winning regions of each agent \cite{isaacs1965differential}. The gray region labelled by $\bar\cX_\mathrm{II}^\mathrm{c}$ ($= \bR^2 \backslash \bar\cX_\mathrm{II}$) is the Defender's winning region; if the Intruder starts in this set, the Defender can capture it by following the equilibrium strategies proposed in, for instance, \cite{lee2021optimal}. The set $\bar\cX_\mathrm{I}$ is the Intruder's winning region in which the Intruder is guaranteed to safely reconnoiter $\cT$ and then return to $\cR$. The white region $\bar\cX_\mathrm{II} \backslash \bar\cX_\mathrm{I}$ is where the Value is negative, where the game may have non-unique solutions.

\begin{figure}[htp]
    \centering
    \includegraphics[scale=0.36]{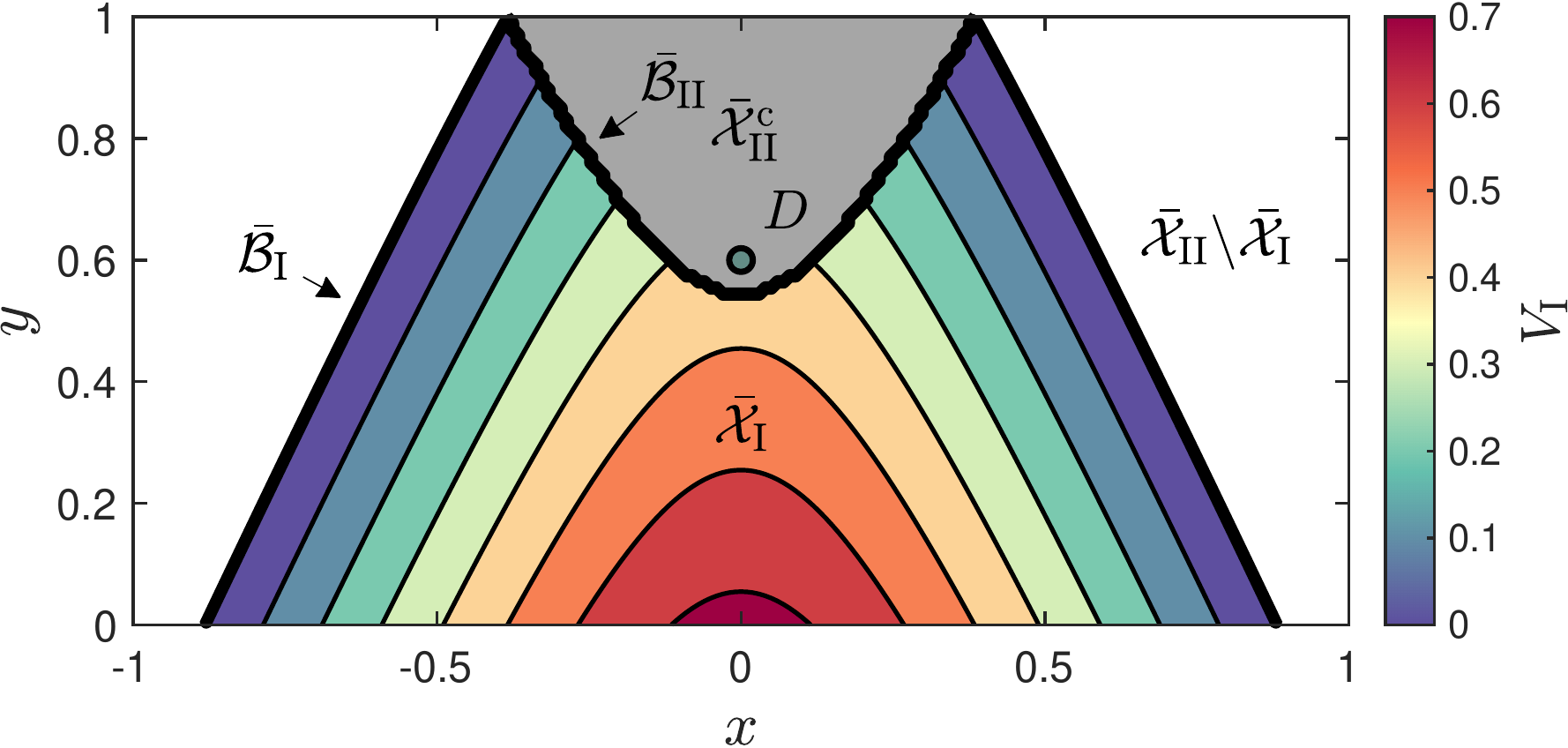}
    \caption{Level sets of $V_\mathrm{I}$ and partitioning of the game state space for the fixed Defender position given in Section \ref{sec:level}.}
    \label{fig:value}
\end{figure}

\section{Simulation Results} \label{sec:simresults}

In this section, three different game scenarios (S1-S3) are examined numerically and are illustrated in Figures \ref{fig:S1}-\ref{fig:S3}. In each subfigure, the filled and unfilled triangles (resp., circles) indicate the positions and initial positions of the Intruder (resp., the Defender) for the corresponding time interval. For simulation parameters, we select $[x_I^0,y_I^0]^\top = [0.5,0.1]^\top$, $[x_D^0,y_D^0]^\top = [0.1,0.9]^\top$, $\alpha=1.2$, and $l = 1$.

In S1, the Intruder (resp., Defender) adheres to its state-feedback equilibrium strategies, namely \eqref{eq:phi_I} for Phase I and \eqref{eq:phi_II} for Phase II (resp., \eqref{eq:psi_I} for Phase I and \eqref{eq:psi_II} for Phase II). See Figure \ref{fig:S1} for the resulting trajectories. As expected in Remark \ref{rem:defopthead}, the Defender's trajectory is a straight line that connects its initial and terminal positions. The Value obtained in this scenario is $0.4682$. 

\begin{figure}[htp!]
    \centering
    \subcaptionbox{$0 \leq t \leq 0.593$ \label{fig:1a}}{
    \begin{tikzpicture}
    \begin{axis}[
        font=\footnotesize,
        scatter/classes={
        a={mark=o,draw=black,thick,scale=0.8},
        b={mark=triangle,draw=black,thick,scale=1},
        c={mark=*,draw=black,thick,fill=teal!50!white,scale=0.8},
        d={mark=triangle*,draw=black,thick,fill=red!50!white,scale=1}
        },
        height=5cm, width=5cm,
        axis on top,
        xlabel=$x$, ylabel=$y$,
        xlabel style={anchor=south},
        ylabel style={anchor=south},
        axis x line=center,
        axis y line=center,
        xmin=0, xmax=1.6,
        ymin=0, ymax=1,
        xtick=\empty, ytick=\empty,
        axis equal
        ]
        \addplot [
        scatter,
        only marks,
        scatter src=explicit symbolic,
        ] coordinates {
        (0.1,0.9) [a]
        (0.5,0.1) [b]
        (0.6874,0.4962) [c]
        (0.9079,0.5318) [d]
        };
        \addplot [
        fill=purple!40!white,
        draw=none
        ] coordinates {
        (-0.2,0) (2,0) (2,-1) (-0.2,-1)
        }; 
        \addplot [
        fill=lightgray!60,
        draw=none
        ] coordinates {
        (-0.2,1) (2,1) (2,2) (-0.2,2)
        };  
        \addplot [
        black,
        domain=-0.2:2,
        samples=10,
        ]
        {1};
        \addplot [
            black,
            solid,
            mark = none,
            postaction={decorate},
            decoration={markings, 
            mark=at position 0.6 with {\arrow{stealth}}},
        ] table {data/x_a_1.dat};
        \addplot [
            black,
            solid,
            mark = none,
            postaction={decorate},
            decoration={markings, 
            mark=at position 0.6 with {\arrow{stealth}}},
        ] table {data/x_d_1.dat};
        \node [anchor=center,font=\small] at (0.8,1.15) {$\cT$}; 
        \node [anchor=center,font=\small] at (0.8,-0.15) {$\cR$}; 
        \node [anchor=north,font=\footnotesize] at (0.1,0.9) {$D$}; 
        \node [anchor=east,font=\footnotesize] at (0.5,0.1) {$I$}; 
        \end{axis}
    \end{tikzpicture}
    }
    ~
    \subcaptionbox{$0.593 \leq t \leq 1.321$ \label{fig:1b}}{
    \begin{tikzpicture}
    \begin{axis}[
        font=\footnotesize,
        scatter/classes={
        a={mark=o,draw=black,thick,scale=0.8},
        b={mark=triangle,draw=black,thick,scale=1},
        c={mark=*,draw=black,thick,fill=teal!50!white,scale=0.8},
        d={mark=triangle*,draw=black,thick,fill=red!50!white,scale=1}
        },
        height=5cm, width=5cm,
        axis on top,
        xlabel=$x$, ylabel=$y$,
        xlabel style={anchor=south},
        ylabel style={anchor=south},
        axis x line=center,
        axis y line=center,
        xmin=0, xmax=1.6,
        ymin=0, ymax=1,
        xtick=\empty, ytick=\empty,
        axis equal
        ]
        \addplot [
        scatter,
        only marks,
        scatter src=explicit symbolic,
        ] coordinates {
        (0.6874,0.4962) [a]
        (0.9079,0.5318) [b]
        (1.4091,0) [d]
        (1.4091,0) [c]
        };
        \addplot [
        fill=purple!40!white,
        draw=none
        ] coordinates {
        (-0.2,0) (2,0) (2,-1) (-0.2,-1)
        }; 
        \addplot [
        fill=lightgray!60,
        draw=none
        ] coordinates {
        (-0.2,1) (2,1) (2,2) (-0.2,2)
        };  
        \addplot [
        black,
        domain=-0.2:2,
        samples=10,
        ]
        {1};
            only marks,
            mark=+,
            mark size=2.5pt,
            draw=orange!60!white,
            very thick
            ]
            coordinates {
            (1.4091,0.6128)
            }; 
        \addplot [
            black,
            solid,
            mark = none,
            postaction={decorate},
            decoration={markings, 
            mark=at position 0.6 with {\arrow{stealth}}},
        ] table {data/x_a_2.dat};
        \addplot [
            black,
            solid,
            mark = none,
            postaction={decorate},
            decoration={markings, 
            mark=at position 0.6 with {\arrow{stealth}}},
        ] table {data/x_d_2.dat};
        \addplot [black,dashed] table {data/x_a_1.dat};
        \addplot [black,dashed] table {data/x_d_1.dat};
        \node [anchor=center,font=\small] at (0.8,1.15) {$\cT$}; 
        \node [anchor=center,font=\small] at (0.8,-0.15) {$\cR$}; 
        \node [anchor=south,font=\footnotesize] at (0.6874,0.4962) {$D$}; 
        \node [anchor=south,font=\footnotesize] at (0.9079,0.5318) {$I$}; 
        \end{axis}
    \end{tikzpicture}
    }
    \caption{ Optimal Intruder vs. optimal Defender (S1) }
    \label{fig:S1}
\end{figure}
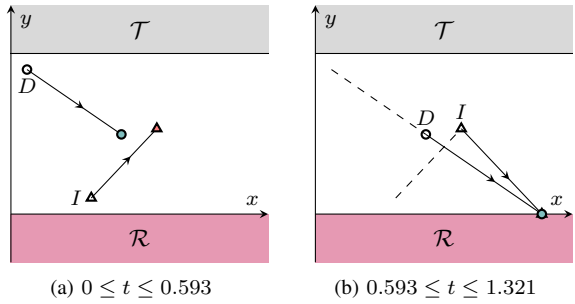

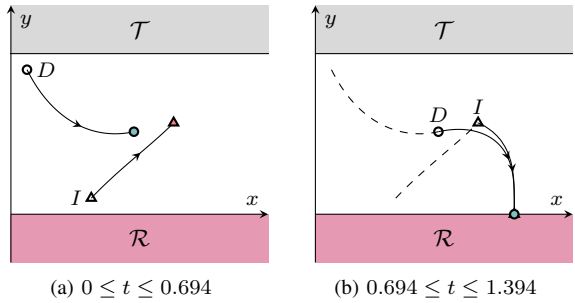
\begin{figure}[htp!]
    \centering
    \subcaptionbox{$0 \leq t \leq 0.694$ \label{fig:1a}}{
    \begin{tikzpicture}
    \begin{axis}[
        font=\footnotesize,
        scatter/classes={
        a={mark=o,draw=black,thick,scale=0.8},
        b={mark=triangle,draw=black,thick,scale=1},
        c={mark=*,draw=black,thick,fill=teal!50!white,scale=0.8},
        d={mark=triangle*,draw=black,thick,fill=red!50!white,scale=1}
        },
        height=5cm, width=5cm,
        axis on top,
        xlabel=$x$, ylabel=$y$,
        xlabel style={anchor=south},
        ylabel style={anchor=south},
        axis x line=center,
        axis y line=center,
        xmin=0, xmax=1.6,
        ymin=0, ymax=1,
        xtick=\empty, ytick=\empty,
        axis equal
        ]
        \addplot [
        scatter,
        only marks,
        scatter src=explicit symbolic,
        ] coordinates {
        (0.1,0.9) [a]
        (0.5,0.1) [b]
        (0.7654,0.5138) [c]
        (1.0126,0.5689) [d]
        };
        \addplot [
        fill=purple!40!white,
        draw=none
        ] coordinates {
        (-0.2,0) (2,0) (2,-1) (-0.2,-1)
        }; 
        \addplot [
        fill=lightgray!60,
        draw=none
        ] coordinates {
        (-0.2,1) (2,1) (2,2) (-0.2,2)
        };  
        \addplot [
        black,
        domain=-0.2:2,
        samples=10,
        ]
        {1};
        \addplot [
            black,
            solid,
            mark = none,
            postaction={decorate},
            decoration={markings, 
            mark=at position 0.6 with {\arrow{stealth}}},
        ] table {data/x_a_3.dat};
        \addplot [
            black,
            solid,
            mark = none,
            postaction={decorate},
            decoration={markings, 
            mark=at position 0.6 with {\arrow{stealth}}},
        ] table {data/x_d_3.dat};
        \node [anchor=center,font=\small] at (0.8,1.15) {$\cT$}; 
        \node [anchor=center,font=\small] at (0.8,-0.15) {$\cR$}; 
        \node [anchor=west,font=\footnotesize] at (0.1,0.9) {$D$}; 
        \node [anchor=east,font=\footnotesize] at (0.5,0.1) {$I$}; 
        \end{axis}
    \end{tikzpicture}
    }
    ~
    \subcaptionbox{$0.694 \leq t \leq 1.394$ \label{fig:1b}}{
    \begin{tikzpicture}
    \begin{axis}[
        font=\footnotesize,
        scatter/classes={
        a={mark=o,draw=black,thick,scale=0.8},
        b={mark=triangle,draw=black,thick,scale=1},
        c={mark=*,draw=black,thick,fill=teal!50!white,scale=0.8},
        d={mark=triangle*,draw=black,thick,fill=red!50!white,scale=1}
        },
        height=5cm, width=5cm,
        axis on top,
        xlabel=$x$, ylabel=$y$,
        xlabel style={anchor=south},
        ylabel style={anchor=south},
        axis x line=center,
        axis y line=center,
        xmin=0, xmax=1.6,
        ymin=0, ymax=1,
        xtick=\empty, ytick=\empty,
        axis equal
        ]
        \addplot [
        scatter,
        only marks,
        scatter src=explicit symbolic,
        ] coordinates {
        (0.7654,0.5138) [a]
        (1.0126,0.5689) [b]
        (1.2377,0) [d]
        (1.2377,0) [c]
        };
        \addplot [
        fill=purple!40!white,
        draw=none
        ] coordinates {
        (-0.2,0) (2,0) (2,-1) (-0.2,-1)
        }; 
        \addplot [
        fill=lightgray!60,
        draw=none
        ] coordinates {
        (-0.2,1) (2,1) (2,2) (-0.2,2)
        };  
        \addplot [
        black,
        domain=-0.2:2,
        samples=10,
        ]
        {1};
        \addplot [
            black,
            solid,
            mark = none,
            postaction={decorate},
            decoration={markings, 
            mark=at position 0.6 with {\arrow{stealth}}},
        ] table {data/x_a_4.dat};
        \addplot [
            black,
            solid,
            mark = none,
            postaction={decorate},
            decoration={markings, 
            mark=at position 0.6 with {\arrow{stealth}}},
        ] table {data/x_d_4.dat};
        \addplot [black,dashed] table {data/x_a_3.dat};
        \addplot [black,dashed] table {data/x_d_3.dat};
        \node [anchor=center,font=\small] at (0.8,1.15) {$\cT$}; 
        \node [anchor=center,font=\small] at (0.8,-0.15) {$\cR$}; 
        \node [anchor=south,font=\footnotesize] at (0.7654,0.5138) {$D$}; 
        \node [anchor=south,font=\footnotesize] at (1.0126,0.5689) {$I$}; 
        \end{axis}
    \end{tikzpicture}
    }
    \caption{ Optimal Intruder vs. nonoptimal Defender (S2) }
    \label{fig:S2}
\end{figure}

\begin{figure}[htp!]
    \centering
    \subcaptionbox{$0 \leq t \leq 0.387$ \label{fig:1a}}{
    \begin{tikzpicture}
    \begin{axis}[
        font=\footnotesize,
        scatter/classes={
        a={mark=o,draw=black,thick,scale=0.8},
        b={mark=triangle,draw=black,thick,scale=1},
        c={mark=*,draw=black,thick,fill=teal!50!white,scale=0.8},
        d={mark=triangle*,draw=black,thick,fill=red!50!white,scale=1}
        },
        height=5cm, width=5cm,
        axis on top,
        xlabel=$x$, ylabel=$y$,
        xlabel style={anchor=south},
        ylabel style={anchor=south},
        axis x line=center,
        axis y line=center,
        xmin=0, xmax=1.6,
        ymin=0, ymax=1,
        xtick=\empty, ytick=\empty,
        axis equal
        ]
        \addplot [
        scatter,
        only marks,
        scatter src=explicit symbolic,
        ] coordinates {
        (0.1,0.9) [a]
        (0.5,0.1) [b]
        (0.4140,0.5630) [c]
        (0.5,0.4880) [d]
        };
        \addplot [
        fill=purple!40!white,
        draw=none
        ] coordinates {
        (-0.2,0) (2,0) (2,-1) (-0.2,-1)
        }; 
        \addplot [
        fill=lightgray!60,
        draw=none
        ] coordinates {
        (-0.2,1) (2,1) (2,2) (-0.2,2)
        };  
        \addplot [
        black,
        domain=-0.2:2,
        samples=10,
        ]
        {1};
        \addplot [
            black,
            solid,
            mark = none,
            postaction={decorate},
            decoration={markings, 
            mark=at position 0.6 with {\arrow{stealth}}},
        ] table {data/x_a_5.dat};
        \addplot [
            black,
            solid,
            mark = none,
            postaction={decorate},
            decoration={markings, 
            mark=at position 0.6 with {\arrow{stealth}}},
        ] table {data/x_d_5.dat};
        \node [anchor=center,font=\small] at (0.8,1.15) {$\cT$}; 
        \node [anchor=center,font=\small] at (0.8,-0.15) {$\cR$}; 
        \node [anchor=north,font=\footnotesize] at (0.1,0.9) {$D$}; 
        \node [anchor=east,font=\footnotesize] at (0.5,0.1) {$I$}; 
        \end{axis}
    \end{tikzpicture}
    }
    ~
    \subcaptionbox{$0.387 \leq t \leq 0.903$ \label{fig:1b}}{
    \begin{tikzpicture}
    \begin{axis}[
        font=\footnotesize,
        scatter/classes={
        a={mark=o,draw=black,thick,scale=0.8},
        b={mark=triangle,draw=black,thick,scale=1},
        c={mark=*,draw=black,thick,fill=teal!50!white,scale=0.8},
        d={mark=triangle*,draw=black,thick,fill=red!50!white,scale=1}
        },
        height=5cm, width=5cm,
        axis on top,
        xlabel=$x$, ylabel=$y$,
        xlabel style={anchor=south},
        ylabel style={anchor=south},
        axis x line=center,
        axis y line=center,
        xmin=0, xmax=1.6,
        ymin=0, ymax=1,
        xtick=\empty, ytick=\empty,
        axis equal
        ]
        \addplot [
        scatter,
        only marks,
        scatter src=explicit symbolic,
        ] coordinates {
        (0.4140,0.5630) [a]
        (0.5,0.4880) [b]
        (0.6954,0) [d]
        (0.6954,0) [c]
        };
        \addplot [
        fill=purple!40!white,
        draw=none
        ] coordinates {
        (-0.2,0) (2,0) (2,-1) (-0.2,-1)
        }; 
        \addplot [
        fill=lightgray!60,
        draw=none
        ] coordinates {
        (-0.2,1) (2,1) (2,2) (-0.2,2)
        };  
        \addplot [
        black,
        domain=-0.2:2,
        samples=10,
        ]
        {1};
        \addplot [
            black,
            solid,
            mark = none,
            postaction={decorate},
            decoration={markings, 
            mark=at position 0.6 with {\arrow{stealth}}},
        ] table {data/x_a_6.dat};
        \addplot [
            black,
            solid,
            mark = none,
            postaction={decorate},
            decoration={markings, 
            mark=at position 0.6 with {\arrow{stealth}}},
        ] table {data/x_d_6.dat};
        \addplot [black,dashed] table {data/x_a_5.dat};
        \addplot [black,dashed] table {data/x_d_5.dat};
        \node [anchor=center,font=\small] at (0.8,1.15) {$\cT$}; 
        \node [anchor=center,font=\small] at (0.8,-0.15) {$\cR$}; 
        \node [anchor=south west,font=\footnotesize] at (0.4140,0.5630) {$D$}; 
        \node [anchor=west,font=\footnotesize] at (0.5,0.4880) {$I$}; 
        \end{axis}
    \end{tikzpicture}
    }
    \caption{ Nonoptimal Intruder vs. optimal Defender (S3) }
    \label{fig:S3}
\end{figure}
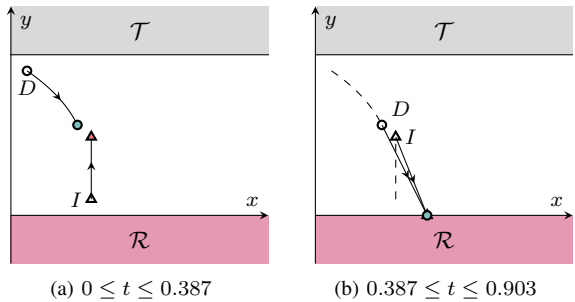

In S2, the Intruder plays optimally (i.e., it employs the equilibrium strategies \eqref{eq:phi_I} and \eqref{eq:phi_II}), whereas the Defender adopts the pure pursuit strategy \cite{isaacs1965differential}. As depicted in Figure \ref{fig:S2}, the Intruder approaches $\cT$ closer than in S1 due to the Defender's nonoptimal play, resulting in the smaller Value, which is $0.4311$. Interestingly, the trajectory of the Intruder during the second phase is not a straight line since chattering occurs (the game alternates between Phases I and II).

In S3, the Defender plays optimally, whereas the Intruder moves directly towards $\cT$ while safe retreat is guaranteed. Figure \ref{fig:S3} shows that the Intruder cannot approach $\cT$ as close as in S1 with the Value of $0.5120$.

Lastly, Figure \ref{fig:values} shows the evolution of the minimum distance between the Intruder and $\cT$ over time. The Intruder acquires the smallest Value in S2, followed by S1 and S3. This shows that the saddle-point condition \eqref{eq:sadpointcond} is satisfied.

\section{Conclusion} \label{sec:conclusions}

In this paper, the two-player reconnaissance game with half-planar target and retreat regions was considered. The game was decomposed into two phases and solved analytically using differential game methods. Simulation results were presented to showcase the performance of the solution.

\begin{figure}
    \centering
    \begin{tikzpicture}
        \begin{axis}[
            legend style={font=\footnotesize},
            major grid style={dotted,black},
            font=\footnotesize,
            width=7cm,
            height=3.5cm,
            xlabel={$t$},
            ylabel={$\min\limits_{\tau \in [0,t]} (l-y_I(\tau))$},
            axis lines = left,
            grid=major,
            legend entries={S1,S2,S3},
            legend pos=north east,
            legend columns=-1,
            ymin=0.4, ymax=0.9,
        ]
        \addplot [
            black!70,
            mark color=black,
            mark=halfcircle*,
            mark repeat={20},
            thick,
        ] table {data/value_S1.dat};
        \addplot [
            red!70,
            mark color=red,
            mark=halfsquare*,
            mark repeat={20},
            thick,
            ] table {data/value_S2.dat};
        \addplot [
            teal!70,
            mark color=teal,
            mark=halfsquare left*,
            mark repeat={20},
            thick,
            ] table {data/value_S3.dat};
        \end{axis}
    \end{tikzpicture}
    \caption{Minimum distance between the Intruder and $\cT$ versus time.}
    \label{fig:values}
\end{figure}
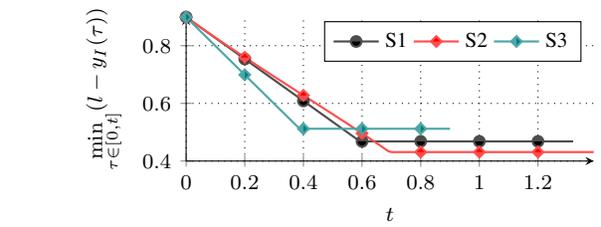


\bibliographystyle{ieeetr}
\bibliography{pegref}

\end{document}